%% file: bds_rap.tex
\documentclass[letterpaper,11pt]{article}

\input{preamble.tex}

\allowdisplaybreaks

\title{RIP-like Properties in Subsampled Blind Deconvolution} 
\author{Kiryung Lee and Marius Junge}

\begin{document}
\doublespacing

\maketitle


\begin{abstract}
  We derive near optimal performance guarantees for subsampled blind deconvolution. Blind deconvolution is an ill-posed bilinear inverse problem and additional subsampling makes the problem even more challenging. Sparsity and spectral flatness priors on unknown signals are introduced to overcome these difficulties. While being crucial for deriving desired near optimal performance guarantees, unlike the sparsity prior with a nice union-of-subspaces structure, the spectral flatness prior corresponds to a nonconvex cone structure, which is not preserved by elementary set operations. This prohibits the operator arising in subsampled blind deconvolution from satisfying the standard restricted isometry property (RIP) at near optimal sample complexity, which motivated us to study other RIP-like properties. Combined with the performance guarantees derived using these RIP-like properties in a companion paper, we show that subsampled blind deconvolution is provably solved at near optimal sample complexity by a practical algorithm.
\end{abstract}

\section{Introduction}

\subsection{Subsampled blind deconvolution of sparse signals}

The subsampled blind deconvolution problem refers to the resolution of two signals from
a few samples of their convolution and is formulated as a bilinear inverse problem as follows.
Let $\Omega = \{\omega_1,\omega_2,\ldots,\omega_m\}$ denote the set of $m$ sampling indices out of $\{1,\ldots,n\}$.
Given $\Omega$, the sampling operator $S_\Omega: \cz^n \to \cz^m$ is defined so that
the $k$th element of $S_\Omega x \in \cz^m$ is the $\omega_k$th element of $x \in \cz^n$ for $k = 1,\ldots,m$.
Then, the $m$ samples of the convolution $x \conv y$ indexed by $\Omega$ with additive noise constitute
the measurement vector $b \in \cz^m$, which is expressed as
\[
b = \sqrt{\frac{n}{m}} S_\Omega (x \conv y) + z,
\]
where $z$ denotes additive noise.

Let $x,y \in \cz^n$ be uniquely represented as $x = \Phi u$ and $y = \Psi v$ over dictionaries $\Phi$ and $\Psi$.
Then, the recovery of $(x,y)$ is equivalent to the recovery of $(u,v)$,
and the subsampled blind deconvolution problem corresponds to the bilinear inverse problem of recovering $(u,v)$
from its bilinear measurements in $b$, when $\Omega$, $\Phi$, and $\Psi$ are known.

A stable reconstruction in subsampled blind deconvolution is defined through the lifting procedure \cite{ahmed2014blind}
that converts the blind deconvolution to recovery of a rank-1 matrix from its linear measurements.
By the lifting procedure, bilinear measurements of $(u,v)$ are equivalently rewritten as linear measurements of the matrix $X = u v^\transpose$,
i.e., there is a linear operator $\A: \cz^{n \times n} \to \cz^m$ such that
\[
b = \A(X) + z.
\]
Then, each element of the measurement vector $b$ corresponds to a matrix inner product.
Indeed, there exist matrices $M_1,M_2,\ldots,M_m \in \cz^{n \times n}$ that describe the action of $\A$ on $X$ by
\begin{equation}
\label{eq:defcalA}
\A(X) = [\langle M_1, X \rangle, \ldots, \langle M_m, X \rangle]^\transpose.
\end{equation}
Since the circular convolution corresponds to the element-wise product in the Fourier domain, $M_\ell$'s are explicitly expressed as
\begin{align*}
M_\ell = \frac{n}{\sqrt{m}} \Phi^* F^* \diag(f_{\omega_\ell}) \overline{F} \overline{\Psi}, \quad \ell = 1,\ldots,m,
\end{align*}
where $f_{\omega_\ell}$ denotes the $\omega_\ell$th column of the unitary DFT matrix $F \in \cz^{n \times n}$.
The subsampled blind deconvolution problem then becomes a matrix-valued linear inverse problem
where the unknown matrix $X$ is constrained to the set of rank-1 matrices.

In the lifted formulation, a reconstruction $\widehat{X}$ of the unknown matrix $X$ is considered successful
if it satisfies the following stability criterion:
\begin{equation}
\label{eq:success}
\frac{\norm{\widehat{X} - X}_{\mathrm{F}}}{\norm{X}_{\mathrm{F}}}
\leq C \left( \frac{\norm{z}_2}{\norm{\A(X)}_2} \right)
\end{equation}
for an absolute constant $C$.
This definition of success is free of the inherent scale ambiguity in the original bilinear formulation.
Once $\widehat{X}$ is recovered, $u$ (resp. $v$) is identified up to a scale factor as the left (resp. right) factor of the rank-1 matrix $\widehat{X}$.

The subsampled blind deconvolution problem is ill-posed and cannot be solved without restrictive models on unknown signals.
We assume the following signal priors, which are modified from a previous subspace model for blind deconvolution \cite{ahmed2014blind}.
\begin{description}

\item[A1] \textbf{Sparsity:} The coefficient vector $u$ is $s_1$-sparse. Geometrically, $u$ belongs to the union of all subspaces spanned by $s_1$ standard basis vectors. The previous subspace model \cite{ahmed2014blind} corresponds to a special case where the subspace in the union that includes $u$ is known a priori. To simplify the notation, define
    \[
    \Gamma_s := \{u \in \cz^n:\pl \norm{u}_0 \leq s \},
    \]
    where $\norm{u}_0$ counts the number of nonzeros in $u$. Then, $u \in \Gamma_{s_1}$.
    The other coefficient vector $v$ is $s_2$-sparse, i.e., $v \in \Gamma_{s_2}$.

\item[A2] \textbf{Spectral flatness:} The unknown signals $x$ and $y$ are flat in the Fourier domain as follows. 
    Define a set $\C_{\mu}$ by
    \begin{equation}
    \label{eq:defCmu}
    \C_{\mu} := \{x \in \cz^n :\pl \textsf{sf}(x) \leq \mu \},
    \end{equation}
    where $\textsf{sf}(x)$ denotes the spectral flatness level of $x \in \cz^n$ given by
    \[
    \textsf{sf}(x) := \frac{n \norm{F x}_\infty^2}{\norm{F x}_2^2}.
    \]
    Then, $x \in \C_{\mu_1}$ and $y \in \C_{\mu_2}$. When $\Phi$ and $\Psi$ are invertible, it is equivalent to $u \in \Phi^{-1} \C_{\mu_1}$ and $v \in \Psi^{-1} \C_{\mu_2}$.\footnote{For simplicity, we restrict our analysis to the case where $\Phi$ and $\Psi$ are invertible matrices. However, it is straightforward to extend the analysis to the case with overcomplete dictionaries by replacing the inverse by the preimage operator.}

\end{description}

Our objective is to show that the subsampled blind deconvolution of signals
following the aforementioned models is possible at near optimal sample complexity.
Similarly to related results in compressed sensing, we take the following two-step approach:
i) First, in a companion paper \cite{LeeLJB2015},
it was shown that stable reconstruction from noisy measurements is available under a restricted isometry property (RIP) of the linear operator $\A$.
In particular, under a mild additional assumption on signals, we show that
a practical algorithm provably achieves stable reconstruction under RIP-like properties of $\A$;
ii) Next, in this paper, we prove that if both dictionaries $\Phi, \Psi \in \cz^{n \times n}$ are mutually independent random matrices whose entries are independent and identically distributed (i.i.d.) following a zero-mean and complex normal distribution $CN(0,1/n)$, with high probability, such RIP-like properties hold at the sample complexity of $m = O(\mu_1 s_2 + \mu_2 s_1) \log^5 n$.
This sample complexity is near optimal (up to a logarithmic factor)
when the spectral flatness parameters $\mu_1$ and $\mu_2$ are sublinear in $s_1$ and $s_2$, respectively;
Combining these results provides the desired near optimal performance guarantees.

\subsection{RIP and RIP-like properties}

We first review RIP and extend the notion to RIP-like properties.
RIP was originally proposed to show the performance guarantee for the recovery in compressed sensing by $\ell_1$-norm minimization \cite{candes2005decoding}.
It is generalized as follows:

\begin{defi}
Let $(\calH,\hsnorm{\cdot})$ be a Hilbert space where $\hsnorm{\cdot}$ denotes the Hilbert-Schmidt norm.
Let $\calS \subset \calH$ be a centered and symmetric set, i.e.,
$0 \in \calS$ and $\alpha \calS = \calS$ for all $\alpha \in \cz$ in the unit modulus.
A linear operator $\A: \calH \to \ell_2^m$ satisfies the $(\calS,\delta)$-RIP if
\[
(1-\delta) \hsnorm{w}^2 \leq \norm{\A(w)}_2^2 \leq (1+\delta) \hsnorm{w}^2, \quad \forall w \in \calS,
\]
or equivalently,
\[
\left| \norm{\A(w)}_2^2 - \hsnorm{w}^2 \right| \leq \delta \hsnorm{w}^2,
\quad \forall w \in \calS.
\]
\end{defi}

Hilbert-Schmidt norms, including the $\ell_2$ norm, are represented as an inner product of a vector with itself.
For example, $\hsnorm{w}^2 = \langle w, w \rangle$ and $\norm{\A(w)}_2^2 = \langle \A(w), \A(w) \rangle$.
This observation extends RIP to another property called \textit{restricted angle-preserving property} (RAP) defined as follows:

\begin{defi}
Let $\calS,\calS' \subset \calH$ be centered and symmetric sets.
A linear operator $\A: \calH \to \ell_2^m$ satisfies the $(\calS,\calT,\delta)$-RAP if
\[
\left| \langle \A(w'), \A(w)\rangle - \langle w', w\rangle \right| \leq \delta \hsnorm{w} \hsnorm{w'},
\quad \forall w \in \calS, w' \in \calS'.
\]
\end{defi}

In a more restrictive case with orthogonality between $w$ and $w'$ ($\langle w', w \rangle = 0$),
RAP reduces to the \textit{restricted orthogonality property} (ROP) \cite{candes2008restricted}.
\begin{defi}
Let $\calM,\calM' \subset \calH$ be centered and symmetric sets.
A linear operator $\A: \calH \to \ell_2^m$ satisfies the $(\calM,\calM',\delta)$-ROP if
\[
\left| \langle \A(w'), \A(w)\rangle \right| \leq \delta \hsnorm{w} \hsnorm{w'},
\quad \forall w \in \calM, \pl \forall w' \in \calM' ~ \text{s.t.} ~ \langle w', w \rangle = 0.
\]
\end{defi}

RIP and RAP of a linear operator $\A$ have useful implications for the inverse problem given by $\A$.
Let $\calS-\calS = \{w-w' :~ w, w' \in \calS\}$.
The $(\calS-\calS,\delta)$-RIP of $\A$ implies that $\A$ is injective when the domain is restricted to $\calS-\calS$;
hence, every $w \in \calS$ is uniquely identified from $\A(w)$.
The $(\calS-\calS,\calS-\calS,\delta)$-RAP was used to show that
practical algorithms, such as the projected gradient method, reconstruct $w$ from $\A(w)$ with a provable performance guarantee.

By definition, the $(\calS,\calS,\delta)$-RAP implies the $(\calS,\delta)$-RIP, but the converse is not true in general.
For certain $\calS$ with special structures, RIP implies RIP-like properties.
For example, when $\calS$ is a subspace, the Minkowski sum of $\calS$ and $-\calS$ coincides with $\calS$.
Therefore, $(\calS,\delta)$-RIP, $(\calS-\calS,\delta)$-RIP, and $(\calS-\calS,\calS-\calS,\delta)$-RAP are all equivalent.
The restrictive set $\calS$ as a subspace arises in many applications.
A set of matrices with Toeplitz, Hankel, circulant, symmetric, or skew symmetric structure corresponds to such an example.

Yet for another example, a sparsity model, which corresponds to a union of subspaces,
provides the desired relationship between RIP and RIP-like properties.
Let $\calS$ be the set $\Gamma_s$ with all $s$-sparse vectors in the Euclidean space.
Then, it follows that the difference set between $\Gamma_s$ and itself is contained within $\Gamma_{2s}$ (another restrictive set of the same structure but with a twice larger parameter), i.e.,
\begin{equation}
\label{eq:inclusion}
\Gamma_s + \Gamma_s \subset \Gamma_{2s}.
\end{equation}
Therefore, we have the following implications:
\begin{itemize}
  \item $(\Gamma_{2s},\delta)$-RIP implies $(\Gamma_s-\Gamma_s,\delta)$-RIP.
  \item $(\Gamma_{3s},\delta)$-RIP implies $(\Gamma_s-\Gamma_s,\Gamma_s,\delta)$-RAP.
  \item $(\Gamma_{4s},\delta)$-RIP implies $(\Gamma_s-\Gamma_s,\Gamma_s-\Gamma_s,\delta)$-RAP.
\end{itemize}
Recall that these RIP-like properties guarantee stable reconstruction of $s$-sparse vectors from $\A(w)$ by practical algorithms.
With the above implications, it suffices to show $(\Gamma_{ks},\delta)$-RIP for $k \in \{2,3,4\}$.
This is why the performance guarantees in compressed sensing are typically given in terms of $(\Gamma_{ks},\delta)$-RIP.
The above argument also applies to an abstract atomic sparsity model \cite{chandrasekaran2012convex}
and to the sparse and rank-1 model \cite{LeeWB2013spf}.

\subsection{RIP-like properties in blind deconvolution}

Next, we present our main results
that derive RIP-like properties of the linear operator $\A$ in subsampled blind deconvolution at near optimal sample complexity.
In fact, these properties hold for a slightly more general model than an exact sparsity model.
To state the main results in this setup, we define a set of approximately $s$-sparse vector by
\begin{equation}
\label{eq:deftGammas}
\tGamma_s := \{ u \in \cz^n :\pl \norm{u}_1 \leq \sqrt{s} \norm{u}_2 \}.
\end{equation}

\begin{theorem}
\label{thm:bds_rap}
There exist absolute numerical constants $C > 0$ and $\beta \in \mathbb{N}$ such that the following holds. 
Let $\Phi, \Psi \in \cz^{n \times n}$ be independent random matrices whose entries are i.i.d. following $CN(0,1/n)$.
Let $\A: \cz^{n \times n} \to \cz^m$ be defined in (\ref{eq:defcalA}).
\begin{enumerate}

\item If $m \geq C \delta^{-2} (s_1 + \mu_1 s_2) \log^5 n$, then with probability at least $1 - n^{-\beta}$,
\begin{align*}
\big| \langle \hat{u} \hat{v}^\transpose, (\A^*\A - \id) (u v^\transpose) \rangle \big| \leq \delta
\norm{\hat{u} \hat{v}^\transpose}_{\mathrm{F}} \norm{u v^\transpose}_{\mathrm{F}}
\end{align*}
for all $u,\hat{u} \in \tGamma_{s_1} \cap \Phi^{-1} \C_{\mu_1}$ and for all $v,\hat{v} \in \tGamma_{s_2}$.

\item If $m \geq C \delta^{-2} (\mu_2 s_1 + s_2) \log^5 n$, then with probability at least $1 - n^{-\beta}$,
\begin{align*}
\big| \langle \hat{u} \hat{v}^\transpose, (\A^*\A - \id) (u v^\transpose) \rangle \big| \leq \delta
\norm{\hat{u} \hat{v}^\transpose}_{\mathrm{F}} \norm{u v^\transpose}_{\mathrm{F}}
\end{align*}
for all $u,\hat{u} \in \tGamma_{s_1}$ and for all $v,\hat{v} \in \tGamma_{s_2} \cap \Psi^{-1} \C_{\mu_2}$.

\end{enumerate}
\end{theorem}

In the course of proving Theorem~\ref{thm:bds_rap}, we also obtain the following corollary,
the proof of which is contained in the proof of Theorem~\ref{thm:bds_rap}.

\begin{cor}
\label{cor:bds_rap}
There exist absolute numerical constants $C > 0$ and $\beta \in \mathbb{N}$ such that the following holds. 
Let $\Phi \in \cz^{n \times n}$ be a random matrix whose entries are i.i.d. following $CN(0,1/n)$. Let $\Psi = I_n$.
Let $\A: \cz^{n \times n} \to \cz^m$ be defined in (\ref{eq:defcalA}).
Suppose that $m \geq C \delta^{-2} (\mu_2 s_1 + s_2) \log^5 n$.
Then, with probability at least $1 - n^{-\beta}$,
\begin{align*}
\big| \langle \hat{u} \hat{v}^\transpose, (\A^*\A - \id) (u v^\transpose) \rangle \big| \leq \delta
\norm{\hat{u} \hat{v}^\transpose}_{\mathrm{F}} \norm{u v^\transpose}_{\mathrm{F}}
\end{align*}
for all $u,\hat{u} \in \tGamma_{s_1}$ and for all $v,\hat{v} \in \tGamma_{s_2} \cap \Psi^{-1} \C_{\mu_2}$.
\end{cor}

\begin{theorem}
\label{thm:bds_rop_cross}
There exist absolute numerical constants $C > 0$ and $\beta \in \mathbb{N}$ such that the following holds. 
Let $\Phi, \Psi \in \cz^{n \times n}$ be independent random matrices whose entries are i.i.d. following $CN(0,1/n)$.
Let $\A: \cz^{n \times n} \to \cz^m$ be defined in (\ref{eq:defcalA}).
If $m \geq C \delta^{-2} (\mu_2 s_1 + \mu_1 s_2) \log^5 n$, then with probability at least $1 - n^{-\beta}$,
\begin{align*}
\big| \langle \hat{u} \hat{v}^\transpose, \A^*\A (u v^\transpose) \rangle \big| \leq \delta
\norm{\hat{u} \hat{v}^\transpose}_{\mathrm{F}} \norm{u v^\transpose}_{\mathrm{F}}
\end{align*}
for all $u \in \tGamma_{s_1}$, $\hat{u} \in \tGamma_{s_1} \cap \C_{\mu_1}$,
$v \in \tGamma_{s_2} \cap \C_{\mu_2}$, and $\hat{v} \in \tGamma_{s_2}$ such that
$\langle u, \hat{u} \rangle = 0$ and $\langle v, \hat{v} \rangle = 0$.
\end{theorem}

\begin{cor}
\label{cor:bds_rop_cross}
There exist absolute numerical constants $C > 0$ and $\beta \in \mathbb{N}$ such that the following holds. 
Let $\Phi, \Psi \in \cz^{n \times n}$ be independent random matrices whose entries are i.i.d. following $CN(0,1/n)$.
Let $\A: \cz^{n \times n} \to \cz^m$ be defined in (\ref{eq:defcalA}).
If $m \geq C \delta^{-2} (\mu_2 s_1 + \mu_1 s_2) \log^5 n$, then with probability at least $1 - n^{-\beta}$,
\begin{equation}
\label{eq:bds_rop}
\big| \langle \hat{u} \hat{v}^\transpose, \A^*\A (u v^\transpose) \rangle \big| \leq 2 \delta
\norm{\hat{u} \hat{v}^\transpose}_{\mathrm{F}} \norm{u v^\transpose}_{\mathrm{F}}
\end{equation}
for all $u \in \tGamma_{s_1}$, $\hat{u} \in \tGamma_{s_1} \cap \C_{\mu_1}$,
$v \in \tGamma_{s_2} \cap \C_{\mu_2}$, and $\hat{v} \in \tGamma_{s_2}$ such that
either $\langle u, \hat{u} \rangle = 0$ or $\langle v, \hat{v} \rangle = 0$.
\end{cor}

\begin{proof}[Proof of Corollary~\ref{cor:bds_rop_cross}]
If suffices to consider the case where $\langle u, \hat{u} \rangle = 0$.
Due to the homogeneity of (\ref{eq:bds_rop}), without loss of generality, we may assume $\norm{v}_2 = \norm{\hat{v}}_2 = 1$.
Decompose $\hat{v}$ as $\hat{v} = P_{R(v)} \hat{v} + P_{R(v)^\perp} \hat{v}$.
Then, $P_{R(v)} \hat{v} = \alpha v$ for $\alpha \in \cz$ satisfying $|\alpha| \leq 1$.
\[
\begin{aligned}
\big| \langle \hat{u} \hat{v}^\transpose, \A^*\A (u v^\transpose) \rangle \big|
{} & \leq \big| \langle \alpha \hat{u} v^\transpose, \A^*\A (u v^\transpose) \rangle \big|
+ \big| \langle \hat{u} (P_{R(v)^\perp} \hat{v})^\transpose, \A^*\A (u v^\transpose) \rangle \big| \\
{} & \leq \delta |\alpha| \norm{\hat{u} \hat{v}^\transpose}_{\mathrm{F}} \norm{u v^\transpose}_{\mathrm{F}}
+ \delta \norm{\hat{u} (P_{R(v)^\perp} \hat{v})^\transpose}_{\mathrm{F}} \norm{u v^\transpose}_{\mathrm{F}} \\
{} & \leq 2 \delta \norm{\hat{u} \hat{v}^\transpose}_{\mathrm{F}} \norm{u v^\transpose}_{\mathrm{F}},
\end{aligned}
\]
where the second step follows from Theorems~\ref{thm:bds_rap} and \ref{thm:bds_rop_cross}.
\end{proof}

The above results combined with their implications in a companion paper \cite{LeeLJB2015} provide
performance guarantees for subsampled blind deconvolution at near optimal sample complexity of $m = O((\mu_1 s_2 + \mu_2 s_1)\log^5 n)$.

Note that Theorem~\ref{thm:bds_rap} derives a sufficient condition respectively
for the $(\widetilde{\calS}_1,\widetilde{\calS}_1,\delta)$-RAP and the $(\widetilde{\calS}_2,\widetilde{\calS}_2,\delta)$-RAP of $\A$,
where $\widetilde{\calS}_1$ and $\widetilde{\calS}_2$ are defined by
\begin{align*}
\widetilde{\calS}_1 {} & := \{ u v^\transpose \in \cz^{n \times n} :\pl u \in \tGamma_{s_1} \cap \Phi^{-1} \C_{\mu_1}, \pl v \in \tGamma_{s_2} \}, \\
\widetilde{\calS}_2 {} & := \{ u v^\transpose \in \cz^{n \times n} :\pl u \in \tGamma_{s_1}, \pl v \in \tGamma_{s_2} \cap \Psi^{-1} \C_{\mu_2} \}.
\end{align*}
On the other hand, Corollary~\ref{cor:bds_rop_cross} derives a sufficient condition for the $(\widetilde{\calS}_1,\widetilde{\calS}_2,2\delta)$-ROP of $\A$.

The derivations of these RIP-like properties are significantly different from the previous RIP analyses in the following senses:
i) In general, a restrictive set does not satisfy an inclusion property like (\ref{eq:inclusion}).
The restrictive sets $\widetilde{\calS}_1$ and $\widetilde{\calS}_2$, induced from both the sparsity and spectral flatness, correspond to this case.
The non-convex cone structure induced from a nonnegativity prior is yet another example for this case.
Therefore, RIP-like properties are not directly implied by the corresponding RIP,
and it is necessary to derive RIP-like properties independently.
ii) More difficulties arise from the subsampling in the time domain following the convolution.
In particular, the random measurement functionals are not mutually independent, which was one of the crucial assumptions in previous RIP analyses.
Technically, deriving the $(\widetilde{\calS}_1,\widetilde{\calS}_2)$-RAP in Theorem~\ref{thm:bds_rop_cross} involves
bonding the deviation of a fourth-order chaos process.
We exploit the total orthogonality assumed in Theorem~\ref{thm:bds_rop_cross} to avoid such a complicated scenario.

Recall that Theorems~\ref{thm:bds_rap} and \ref{thm:bds_rop_cross} consider an approximate sparsity model
that covers a wider set $\tGamma_s$ than the set $\Gamma_s$ of exactly $s$-sparse vectors.
During the proofs, we also provide extensions of conventional RIP analysis of an i.i.d. subgaussian sensing matrix and partial Fourier sensing matrix in compressed sensing as side results, which might be of independent interest.

The rest of this paper is organized as follows:
In Section~\ref{sec:chaos}, we extend the previous work on suprema of chaos processes by Krahmer et al. \cite{krahmer2014suprema}
from a quadratic form to a bilinear form.
Key entropy estimates are derived in Section~\ref{sec:key_est} along with
their applications to showing the RIP of random matrices for approximately sparse vectors.
In Section~\ref{sec:proofs_main}, the proofs for the main theorems are presented.
Then, we conclude the paper with discussions.

\subsection{Notations}

Various norms are used in this paper.
The Frobenius norm of a matrix is denoted by $\norm{\cdot}_{\mathrm{F}}$.
The operator norm from $\ell_p^n$ to $\ell_q^n$ will be $\norm{\cdot}_{p\to q}$.
Absolute constants will be used throughout the paper.
Symbols $C,c_1,c_2,\ldots$ are reserved for real-valued positive absolute constants.
Symbols $\beta \in \mathbb{N}$ is a positive integer absolute constant.
For a matrix $A$, its element-wise complex conjugate, its transpose, and its Hermitian transpose
are respectively written as $\overline{A}$, $A^\transpose$, and $A^*$.
For a linear operator $\A$ between two vector spaces, $\A$ will denote its adjoint operator.
The matrix inner product $\trace(A^*B)$ between two matrices $A$ and $B$ is denoted by $\langle A, B \rangle$.
Matrix $F \in \cz^{n \times n}$ will represent the unitary discrete Fourier transform
and $\conv$ stands for the circular convolution where its length is clear from the context.
We will use the shorthand notation $[n] = \{1,2, \ldots, n\}$.
Let $J \subset [n]$. Then, $\Pi_J: \cz^n \to \cz^n$ denotes the coordinate projection
whose action on a vector $x$ keeps the entries of $x$ indexed by $J$ and sets the remaining entries to zero.
The identity map on $\cz^{n \times n}$ will be denoted by $\id$.

\section{Suprema of Chaos Processes}
\label{sec:chaos}

\subsection{Covering number and dyadic entropy number}

Let $B, D \subset X$ be convex sets where $X$ is a Banach space.
The $\epsilon$-covering number, denoted by $N(B,\epsilon D)$, is defined as
\[
N(B, \epsilon D) := \inf
\left\{ k \in \mathbb{N} \pl | \pl \exists (x_i)_{i=1}^k \subset X \pl \text{s.t.} \pl B \subset \bigcup_{i=1}^k x_i+\epsilon D \right\}.
\]
The $k$th dyadic entropy number, denoted by $e_k(B, D)$, is defined as
\[
e_k(B, D) := \inf
\left\{\epsilon > 0 \pl \big| \pl \exists (x_i)_{i=1}^{2^{k-1}} \subset X \pl \text{s.t.} \pl B \subset \bigcup_{i=1}^{2^{k-1}}x_i+\epsilon D
\right\}.
\]

Then, the covering number and dyadic entropy number satisfy
\begin{equation}
\label{eq:intN2sumek}
\int_0^\infty \sqrt{ \log N(B,\epsilon D) } \pl d\epsilon
\lesssim \sum_{k=1}^\infty \frac{e_k(B, D)}{\sqrt{k}}.
\end{equation}

Indeed, the inequality in (\ref{eq:intN2sumek}) is derived as follows:
\begin{align*}
{} & \int_0^\infty \sqrt{ \log N(B,\epsilon D) } \pl d\epsilon \\
{} & = \sum_{k=1}^\infty \int_{e_k(B, D)}^{e_{k+1}(B, D)} \sqrt{ \log N(B,\epsilon D) } \pl d\epsilon \\
{} & \leq \sum_{k=1}^\infty \int_{e_k(B, D)}^{e_{k+1}(B, D)} \sqrt{\log 2} \sqrt{k} \pl d\epsilon \\
{} & = \sqrt{\log 2} \sum_{k=1}^\infty \left[e_k(B, D) - e_{k+1}(B, D)\right] \sqrt{k} \\
{} & = \sqrt{\log 2} \sum_{k=1}^\infty (\sqrt{k} - \sqrt{k-1}) e_k(B, D) \\
{} & \leq \sqrt{\log 2} \sum_{k=1}^\infty \frac{e_k(B, D)}{\sqrt{k}}.
\end{align*}

\subsection{Subadditivity of $\gamma_2$ functional}

Let $(T,d)$ be a metric space.
An admissible sequence of $T$, denoted by $\{T_r\}_{r=0}^\infty$,
is a collection of subsets of $T$ that satisfies $|T_0| = 1$ and $|T_r| \leq 2^{2^r}$ for all $r \geq 1$.
The $\gamma_2$ functional \cite{talagrand2005generic} is defined by
\[
\gamma_2(T,d) := \inf_{\{T_r\}} \sup_{t \in T} \sum_{r=0}^\infty 2^{r/2} d(t,T_r) \pl.
\]

\begin{lemma}
\label{lemma:sa_gamma2}
Let $(T,d)$ and $(S,d)$ be metric spaces embedded in a common vector space. Then,
\[
\gamma_2(T+S,d) \leq (1+\sqrt{2}) (\gamma_2(T,d) + \gamma_2(S,d)) \pl.
\]
\end{lemma}
\begin{proof}
Let $\{T_r\}_{r=0}^\infty$ and $\{S_r\}_{r=0}^\infty$ denote admissible sequences for $T$ and $S$, respectively.
Define $\{R_r\}_{r=0}^\infty$ by $R_0 = T_0 + S_0$ and $R_r = T_{r-1} + S_{r-1}$ for $r \geq 1$.
Then, $R_r \subset T+S$ for all $r \geq 0$, and $\{R_r\}_{r=0}^\infty$ satisfies
$|R_0| = 1$ and $|R_r| = |T_{r-1}| |S_{r-1}| \leq 2^{2^{r-1}} 2^{2^{r-1}} = 2^{2^r}$ for all $r \geq 1$.
This implies that $\{R_r\}_{r=0}^\infty$ is an admissible sequence of $T+S$.
By the definition of the $\gamma_2$ functional, we have
\begin{align*}
\gamma_2(T+S,d)
{} & \leq \sup_{t \in T, s \in S} \sum_{r=0}^\infty 2^{r/2} d(t+s,R_r) \\
{} & = \sup_{t \in T, s \in S} d(t+s,T_0+S_0) + \sum_{r=1}^\infty 2^{r/2} d(t+s,T_{r-1}+S_{r-1}) \\
{} & = \sup_{t \in T, s \in S} (1+\sqrt{2}) \sum_{r=1}^\infty 2^{(r-1)/2} d(t+s,T_{r-1}+S_{r-1}) \\
{} & \leq \sup_{t \in T, s \in S} (1+\sqrt{2}) \sum_{r=1}^\infty 2^{(r-1)/2} \left\{ d(t,T_{r-1}) + d(s,S_{r-1}) \right\} \\
{} & = (1+\sqrt{2}) \left\{ \sup_{t \in T} \sum_{r=0}^\infty 2^{r/2} d(t,T_r) + \sup_{s \in S} \sum_{r=0}^\infty 2^{r/2} d(s,S_r) \right\} \pl,
\end{align*}
where the second inequality holds because the metric $d$ satisfies the triangle inequality.
Since the choice of admissible sequences $\{T_r\}_{r=0}^\infty$ and $\{S_r\}_{r=0}^\infty$ was arbitrary,
by taking the infimum with respect to $\{T_r\}_{r=0}^\infty$ and $\{S_r\}_{r=0}^\infty$, we get the desired inequality.
\end{proof}

\subsection{Suprema of chaos processes: bilinear forms}

Krahmer et al. \cite{krahmer2014suprema} showed the concentration of a subgaussian quadratic form.

\begin{theorem}[{\cite[Theorem~3.1]{krahmer2014suprema}}]
\label{thm:kmr}
Let $\xi \in \cz^n$ be an $L$-subgaussian vector with $\mathbb{E} \xi \xi^* = I_n$.
Let $\Delta \subset \cz^{m \times n}$.
Then for $t > 0$,
\[
\mathbb{P}\left(
\sup_{M \in \Delta} \left| \norm{M \xi}_2^2 - \mathbb{E} \norm{M \xi}_2^2 \right|
\geq c_1 K_1(\Delta) + t \pl
\right)
\leq 2 \exp
\left(
-c_2
\min
\left\{
\frac{t^2}{[K_2(\Delta)]^2}, \frac{t}{K_3(\Delta)}
\right\}
\right) \pl,
\]
where $c_1$ and $c_2$ are constants that only depend on $L$, and $K_1$, $K_2$, and $K_3$ are given by
\begin{align*}
K_1(\Delta) {} & := \gamma_2(\Delta, \norm{\cdot}_{2\to2})
\left[ \gamma_2(\Delta, \norm{\cdot}_{2\to2}) + d_{\mathrm{F}}(\Delta) \right]
+ d_{\mathrm{F}}(\Delta) d_{2\to2}(\Delta), \\
K_2(\Delta) {} & := d_{2\to2}(\Delta) \left[ \gamma_2(\Delta, \norm{\cdot}_{2\to2}) + d_{\mathrm{F}}(\Delta) \right], \\
K_3(\Delta) {} & := d_{2\to2}^2(\Delta).
\end{align*}
\end{theorem}

Our main observation here is that a simple application of the polarization identity provides
the extension of the concentration result by Krahmer et al. \cite{krahmer2014suprema}
from a subgaussian quadratic form to a subgaussian bilinear form.
Note that a quadratic form is a special case of a bilinear form.

\begin{theorem}
\label{thm:ip}
Let $\xi \in \cz^n$ be an $L$-subgaussian vector with $\mathbb{E} \xi \xi^* = I_n$.
Let $\Delta,\Delta' \subset \cz^{m \times n}$.
Then for $t > 0$,
\begin{align*}
{} & \mathbb{P}
\Bigg(
\sup_{M \in \Delta, M' \in \Delta'} \big| \langle M' \xi, M \xi \rangle - \mathbb{E} \langle M' \xi, M \xi \rangle \big|
\geq c_1 \max\{K_1(\Delta), K_1(\Delta')\} + t \pl
\Bigg) \\
{} & \quad \leq 8 \exp
\Bigg(
-c_2
\min
\Bigg\{
\frac{t^2}{[\max\{K_2(\Delta), K_2(\Delta')\}]^2}, \frac{t}{\max\{K_3(\Delta), K_3(\Delta')\}}
\Bigg\}
\Bigg) \pl,
\end{align*}
where $c_1$ and $c_2$ are constants that only depend on $L$, and $K_1$, $K_2$, and $K_3$ are defined in Theorem~\ref{thm:kmr}.
\end{theorem}

\begin{proof}[Proof of Theorem~\ref{thm:ip}]
The main result in \cite[Theorem~3.5]{krahmer2014suprema} states that for a collection of self-adjoint matrices $\Delta$
\begin{equation}
\label{KMR}
(\ez \sup_{M \in \Delta}  | \norm{M \xi}^2 - \ez \norm{M \xi}^2 |^p)^{1/p} \lesssim \textup{\textsf{rs}}_p(\Delta) \pl ,
\end{equation}
where the terms $\textup{\textsf{rs}}_p(\Delta)$ is defined by
\begin{equation}
\label{eq:defrs}
\begin{aligned}
\textup{\textsf{rs}}_p(\Delta) :=
{} & \gamma_2(\Delta,\norm{\cdot}_{2\to2}) (\gamma_2(\Delta,\norm{\cdot}_{2\to2}) + d_{\mathrm{F}}(\Delta)) \\
{} & \quad + \sqrt{p} d_{2\to2}(\Delta) (\gamma_2(\Delta,\norm{\cdot}_{2\to2}) + d_{\mathrm{F}}(\Delta))
+ p d_{2\to2}^2(\Delta) \pl .
\end{aligned}
\end{equation}

By the polarization identity and the subadditivity of $\textup{\textsf{rs}}_p(\Delta)$ with respect to the Minkowski sum (Lemma~\ref{lemma:rs}),
we extend \cite[Theorem~3.5]{krahmer2014suprema} to the bilinear case, which is summarized in Lemma~\ref{lemma:ip}.

The next step of applying Markov's inequality to the $p$th moment in the proof of Theorem~\ref{thm:kmr} applies here without modification,
which competes the proof.
\end{proof}

\begin{lemma}
\label{lemma:rs}
Let $\textup{\textsf{rs}}_p$ be as defined in (\ref{eq:defrs}).
For every complex number $\alpha$ of unit modulus,
\[
\textup{\textsf{rs}}_p(\Delta + \alpha \Delta') \lesssim \max(\textup{\textsf{rs}}_p(\Delta), \textup{\textsf{rs}}_p(\Delta')) \pl.
\]
\end{lemma}

\begin{proof}
By the triangle inequality, we have
$d_{2\to2}(\Delta+\alpha\Delta') \leq d_{2\to2}(\Delta) + d_{2\to2}(\Delta')$
and $d_{\mathrm{F}}(\Delta+\alpha\Delta') \leq d_{\mathrm{F}}(\Delta) + d_{\mathrm{F}}(\Delta')$.
Moreover, Lemma~\ref{lemma:sa_gamma2} implies
\begin{align*}
{} &\gamma_2(\Delta+\alpha\Delta',\norm{\cdot}_{2\to2}) \\
{} & \leq (1+\sqrt{2}) \left\{\gamma_2(\Delta,\norm{\cdot}_{2\to2}) + \gamma_2(\alpha\Delta',\norm{\cdot}_{2\to2})\right\} \\
{} & = (1+\sqrt{2}) \left\{\gamma_2(\Delta,\norm{\cdot}_{2\to2}) + \gamma_2(\Delta',\norm{\cdot}_{2\to2})\right\} \pl.
\end{align*}
The assertion follows by applying these results to the definition of $\textup{\textsf{rs}}_p$.
\end{proof}

\begin{lemma}
\label{lemma:ip}
Let $\xi \in \cz^n$ be an $L$-subgaussian vector with $\mathbb{E} \xi \xi^* = I_n$.
Let $\Delta, \Delta' \subset \cz^{n \times n}$.
Then for every $p \geq 1$,
\begin{align*}
\left(\ez \sup_{M \in \Delta, M' \in \Delta'} | \langle M' \xi, M \xi \rangle - \ez \langle M' \xi, M \xi \rangle |^p\right)^{1/p}
\lesssim_L \max(\textup{\textsf{rs}}_p(\Delta), \textup{\textsf{rs}}_p(\Delta')) \pl .
\end{align*}
\end{lemma}

\begin{proof}[Proof of Lemma~\ref{lemma:ip}]
By the polarization identity, we have
\begin{align*}
& \left| \langle M' \xi, M \xi \rangle - \ez \langle M' \xi, M \xi \rangle \right| \\
& = \frac{1}{4} \pl \Big| \sum_{\alpha \in \{\pm 1, \pm i\}} \alpha \left[(M \xi + \alpha M' \xi, M \xi + \alpha M' \xi) - \ez (M \xi + \alpha M' \xi, M \xi + \alpha M' \xi)\right] \Big| \\
& \le \frac{1}{4} \sum_{\alpha \in \{\pm 1, \pm i\}} \left| \p \norm{(M + \alpha M') \xi}_2^2 - \ez \norm{(M + \alpha M') \xi}_2^2 \right|
\end{align*}
Now the triangle inequality in $L_p$ (for $p\gl 1$) implies the assertion in combination with Lemma~\ref{lemma:rs}.
\end{proof}

\section{Key Entropy Estimates}
\label{sec:key_est}

In this section, we derive entropy estimates (lemmas~\ref{lemma:est_id} and \ref{lemma:est_fourier}),
which are key components in the proofs of the main results in Section~\ref{sec:proofs_main}.
These lemmas also extend the previous RIP results on certain random matrices
to the case where the linear operator is restricted to the set of compressible vectors instead of exactly sparse vectors.

The restricted isometry property of a subgaussian matrix and a partial Fourier matrix has been well studied in the compressed sensing literature.
The restrictive model in these studies was the standard sparsity model,
which consists of exactly $s$-sparse vectors in $\Gamma_s$.
We will derive Lemmas~\ref{lemma:est_id} and \ref{lemma:est_fourier}
in the course of extending the previously known $(\Gamma_s,\delta)$-RIP of random matrices to the $(\tGamma_s,\delta)$-RIP,
where the set of approximately $s$-sparse vectors $\tGamma_s$ is defined in (\ref{eq:deftGammas}).

\subsection{Subgaussian linear operator}

We start with a subgaussian matrix $A \in \mathbb{R}^{m \times n}$, whose entries are i.i.d. following $N(0,1/m)$.
Several derivations of the $(\Gamma_s,\delta)$-RIP of $A$ have been presented (cf. \cite{candes2005decoding,baraniuk2008simple,krahmer2014suprema}).
For example, the recent result by Krahmer et al. \cite{krahmer2014suprema} is summarized as follows:

\begin{theorem}[{\cite[Theorem~C.1]{krahmer2014suprema}}]
\label{thm:gauArip_exact}
A subgaussian matrix $A \in \mathbb{R}^{m \times n}$ satisfies $(\Gamma_s,\delta)$-RIP with probability at least $1-\epsilon$ if
\[
m \geq C \delta^{-2} \max\{s\log(en/s),\log(\epsilon^{-1})\}.
\]
\end{theorem}

Earlier proofs \cite{candes2005decoding,baraniuk2008simple} consist of the following two steps:
i) For any $J \subset \{1,\ldots,n\}$ with $|J| = s$,
the corresponding submatrix $A_J$, with columns of $A$ indexed by $J$,
has its singular values concentrated within $(1-\delta,1+\delta)$ except with exponentially small probability;
ii) An upper bound on the probability for the violation ($\norm{A_J^* A_J - I_s} > \delta$) with the worst case choice of $J$,
obtained by a union bound, still remains small.
The first step was shown either by the large deviation result \cite{davidson2001local}
or by a standard volume argument together with the concentration of a subgaussian quadratic form.
It is not straightforward to extend these approaches to the case where the restriction set includes approximately $s$-sparse vectors.
Recently, Krahmer et al. \cite[Appendix~C]{krahmer2014suprema} proposed an alternative derivation of the $(\Gamma_s,\delta)$-RIP of a subgaussian matrix $A$.
They derived a Dudley-type upper bound on the $\gamma_2$ function of $B_2^n \cap \Gamma_s$ (the set of $s$-sparse vectors within the unit $\ell_2$ ball) given by
\begin{equation}
\label{eq:ubgamma2kmo}
\int_0^\infty \sqrt{ \log N(B_2^n \cap \Gamma_s,\epsilon B_2^n) } d\epsilon
\lesssim \sqrt{s \log(en/s)}.
\end{equation}
We extend their result in (\ref{eq:ubgamma2kmo}) to the approximately sparse case, which is stated in the following lemma.
\begin{lemma}
\label{lemma:est_id}
\[
\int_0^\infty \sqrt{ \log N(B_2^n \cap \tGamma_s,\epsilon B_2^n) } d\epsilon
\lesssim \sqrt{s} \log^{3/2} n.
\]
\end{lemma}

\begin{rem}
Lemma~\ref{lemma:est_id} provides an upper bound of the $\gamma_2$ function of a larger set $B_2^n \cap \tGamma_s$,
consisting of approximately $s$-sparse vectors, instead of the set $B_2^n \cap \Gamma_s$ of exactly $s$-sparse unit vectors.
On the other hand, unlike the upper bound in (\ref{eq:ubgamma2kmo}),
the bound in Lemma~\ref{lemma:est_id} is suboptimal, but only by a logarithmic factor.
\end{rem}

\begin{proof}[Proof of Lemma~\ref{lemma:est_id}]
Since $\tGamma_s \cap B_2^n \subset \sqrt{s} B_1^n$, we have
\begin{equation}
\label{eq:sumek2}
\begin{aligned}
{} & \int_0^\infty \sqrt{ \log N(B_2^n \cap \tGamma_s,\epsilon B_2^n) } d\epsilon \\
{} & \leq \int_0^\infty \sqrt{ \log N(\sqrt{s} B_1^n, \epsilon B_2^n) } d\epsilon \\
{} & = \sqrt{s} \int_0^\infty \sqrt{ \log N(B_1^n, \epsilon B_2^n) } d\epsilon \\
{} & \leq \sqrt{s} \sum_{k=1}^{\infty} \frac{e_k( B_1^n, B_2^n)}{\sqrt{k}},
\end{aligned}
\end{equation}
where the second step holds by the change of variables, and the third step follows from (\ref{eq:intN2sumek}).

Note that $\ell_p^n$ is of type-$p$ if $1 \leq p \leq 2$ and of type-2 if $p > 2$.
Furthermore, $I_n: \ell_1^n \to \ell_p^n$ is a contraction,
Therefore, Maurey's empirical method (cf. \cite[Proposition~2]{carl1985inequalities}, \cite{schutt1984entropy}) implies
\[
e_k(B_1^n, B_p^n) \lesssim \sqrt{p} f(k,n,\min(2,p)),
\]
where $f(k,n,p)$ is defined by
\[
f(k,n,p) :=
2^{-\max(k/n,1)} \min\left\{1 , \max\left[ \frac{\log(n/k+1)}{k} , \frac{1}{n} \right]^{1-1/p} \right\}.
\]

Let $a > 0$ denote the unique solution to $\log(a+1) = 1/a$. Then, $a > 1$.
The following cases for $n/k$ cover all possible scenarios.
\begin{description}
  \item[Case 1:] If $n/k > a$, then
  \[
  f(k,n,2) \leq 2^{-1} \sqrt{\frac{\log(n/k+1)}{k}}.
  \]

  \item[Case 2:] If $1 < n/k \leq a$, then
  \[
  f(k,n,2) = \frac{1}{2 \sqrt{n}} < \frac{1}{2 \sqrt{k}}.
  \]

  \item[Case 3:] If $n/k \leq 1$, then since $2^{-k/n} \leq \sqrt{n/k}$ for $k \geq n$, we have
  \[
  f(k,n,2) = 2^{-k/n} \frac{1}{\sqrt{n}} \leq \frac{1}{\sqrt{k}}.
  \]
\end{description}
Therefore,
\[
f(k,n,2) \lesssim \sqrt{\frac{\log (1+n/k)}{k}} \lesssim \sqrt{\frac{\log n}{k}},
\]
which implies
\begin{equation}
\label{eq:sumek2a}
\sum_{k=1}^{n^2-1} \frac{e_k(B_1^n, B_2^n)}{\sqrt{k}}
\lesssim \sum_{k=1}^{n^2-1} \frac{ \sqrt{\log n} }{ k } \leq \log^{3/2} n.
\end{equation}

For $k \geq n^2$, we use the standard volume argument to get
\[
e_k(B_2^n, B_2^n) \leq n/k.
\]
Indeed, by the standard volume argument (\cite[Lemma~1.7]{pisier1206probabilistic}), we have
\[
N(B_2^n,\epsilon B_2^n) \leq (1+2/\epsilon)^n \leq (3/\epsilon)^n,
\]
which implies
\[
e_k(B_2^n, B_2^n) \leq 3 \cdot 2^{-(k-1)/n} \leq 2^{-k/(2n)} \leq n/k.
\]

Therefore,
\begin{equation}
\label{eq:sumek2b}
\sum_{k=n^2}^{\infty} \frac{e_k(B_1^n, B_2^n)}{\sqrt{k}}
\leq \sum_{k=n^2}^{\infty} \frac{e_k(B_2^n, B_2^n)}{\sqrt{k}}
\leq \sum_{k=n^2}^{\infty} \frac{ n }{ k^{3/2} } \leq 2,
\end{equation}
where the first step holds since $B_1^n \subset B_2^n$.

Applying (\ref{eq:sumek2a}) and (\ref{eq:sumek2b}) to (\ref{eq:sumek2}) completes the proof.
\end{proof}

By replacing (\ref{eq:ubgamma2kmo}) in the proof of \cite[Theorem~C.1]{krahmer2014suprema} by Lemma~\ref{lemma:est_id},
we obtain the following theorem that gives the $(\tGamma_s,\delta)$-RIP of a subgaussian matrix.

\begin{theorem}
\label{thm:gauArip}
A subgaussian matrix $A \in \mathbb{R}^{m \times n}$ satisfies $(\tGamma_s,\delta)$-RIP with probability at least $1-\epsilon$ if
\[
m \geq C \delta^{-2} \max\{s\log^3 n,\log(\epsilon^{-1})\}.
\]
\end{theorem}

\subsection{Randomly sampled Fourier transform}

The $(\Gamma_s,\delta)$-RIP of a partial Fourier matrix at near optimal sample complexity was shown \cite{candes2006near,rudelson2008sparse}.
The result further generalized to randomly sampled frame operators \cite{rauhut2010compressive}.
Similarly to the previous section, we will extend a key entropy estimate in previous works \cite{rudelson2008sparse,rauhut2010compressive} from the set $\Gamma_s$ to its superset $\tGamma_s$.

Let $T: \cz^n \to \cz^n$ be a unitary transform so that $T^* T = T T^* = I_n$.
Let $\Omega = \{\omega_1,\omega_2,\ldots,\omega_m\} \subset \{1,\ldots,n\}$ denote the set of $m$ sampling indices.
Given $\Omega$, the sampling operator $S_\Omega: \cz^n \to \cz^m$ is defined so that
the $k$th element of $S_\Omega x \in \cz^m$ is the $\omega_k$th element of $x \in \cz^n$ for $k = 1,\ldots,m$.

\begin{theorem}[{\cite[Theorem~3.3]{rudelson2008sparse}\footnote{A slightly different assumption on $\Omega$ is used in \cite{rudelson2008sparse}. But the result and its proof remain intact with the change.}},{\cite[Theorem~4.4]{rauhut2010compressive}}]
\label{thm:fouArip}
Suppose that $(\omega_k)_{k=1}^m$ be i.i.d. following the uniform distribution on $\{1,\ldots,n\}$.
A random matrix $A \in \mathbb{R}^{m \times n}$ constructed by
\[
A = \sqrt{\frac{n}{m}} S_\Omega T,
\]
satisfies $(\Gamma_s,\delta)$-RIP with probability at least $1-n^{-\beta}$ if $m \geq C \delta^{-2} s \log^5 n$
for absolute constants $C, \beta > 0$.
\end{theorem}

One of the key steps in the proof of Theorem~\ref{thm:fouArip} involves the entropy estimate in the following inequality (a paraphrased version of \cite[Eq.~(13)]{rudelson2008sparse}):
Conditioned on $\Omega$, we have
\begin{equation}
\label{eq:ubNrv}
\int_0^\infty \sqrt{ \log N(S_\Omega T(B_2^n \cap \Gamma_s),\epsilon B_{\infty}^m) } d\epsilon
\lesssim \norm{T}_{1 \to \infty} \sqrt{s} \log s \log^{1/2}m \log^{1/2}n.
\end{equation}

We extend this result to the analogous entropy estimate for $\tGamma_s$ in the following Lemma.

\begin{lemma}
\label{lemma:est_fourier}
Let $T: \cz^n \to \cz^m$ where $m \leq n$. Then,
\[
\int_0^\infty \sqrt{ \log N(S_\Omega T(B_2^n \cap \tGamma_s),\epsilon B_{\infty}^m) } d\epsilon
\lesssim \norm{T}_{1 \to \infty} \sqrt{s} \log^{1/2}m \log^{3/2}n.
\]
\end{lemma}

While applying to a larger set $\tGamma_s$,
the upper bound in Lemma~\ref{lemma:est_fourier} is larger than that of (\ref{eq:ubNrv}) only by a logarithmic factor of $\log n / \log s$.

Replacing (\ref{eq:ubNrv}) in the proof of Theorem~\ref{thm:fouArip} \cite{rudelson2008sparse} by Lemma~\ref{lemma:est_fourier}
extends the RIP result in Theorem~\ref{thm:fouArip} to the compressible case as follows:

\begin{theorem}
Suppose that $(\omega_k)_{k=1}^m$ be i.i.d. following the uniform distribution on $\{1,\ldots,n\}$.
A random matrix $A \in \mathbb{R}^{m \times n}$ constructed by
\[
A = \sqrt{\frac{n}{m}} S_\Omega T,
\]
satisfies $(\tGamma_s,\delta)$-RIP with probability at least $1-n^{-\beta}$ if $m \geq C \delta^{-2} s \log^5 n$
for absolute constants $C, \beta > 0$.
\end{theorem}

The proof of Lemma~\ref{lemma:est_fourier} is given below.

\begin{proof}[Proof of Lemma~\ref{lemma:est_fourier}]
Since $\tGamma_s \cap B_2^n \subset \sqrt{s} B_1^n$, we have
\begin{equation}
\label{eq:sumek1}
\begin{aligned}
{} & \int_0^\infty \sqrt{ \log N(S_\Omega T(B_2^n \cap \tGamma_s), \epsilon B_{\infty}^n) } d\epsilon \\
{} & \leq \int_0^\infty \sqrt{ \log N(\sqrt{s} S_\Omega T(B_1^n), \epsilon B_{\infty}^n) } d\epsilon \\
{} & \leq \sqrt{s} \int_0^\infty \sqrt{ \log N(S_\Omega T(B_1^n), \epsilon B_{\infty}^n) } d\epsilon \\
{} & \lesssim \sqrt{s} \sum_{k=1}^{\infty} \frac{e_k( S_\Omega T(B_1^n), B_{\infty}^n)}{\sqrt{k}},
\end{aligned}
\end{equation}
where the last inequality follows from (\ref{eq:intN2sumek}).

Maurey's empirical method \cite[Proposition~3]{carl1985inequalities} implies
\[
e_k(S_\Omega T(B_1^n), B_{\infty}^m) \lesssim \norm{S_\Omega T}_{1 \to \infty}
h(k,n,m),
\]
where $h(k,n,m)$ is defined as
\begin{align*}
h(k,n,m) {} & :=
2^{-\max(k/n,k/m,1)} \max\left[1,\log^{1/2}(m/k+1)\right] \\
{} & \quad \cdot \min\left\{1 , \max\left[ \frac{\log(n/k+1)}{k} , \frac{1}{n} \right]^{1/2} \right\}.
\end{align*}

Let $a > 0$ denote the unique solution to $\log(a+1) = 1/a$. Then, $a > 1$.
Then, it suffices to consider the following three cases $n/k$.

\begin{description}
  \item[Case 1:] If $n/k > a$, then
  \[
  h(k,n,m) \leq 2^{-1} \sqrt{\frac{\log(m/k+1) \log(n/k+1)}{k}}.
  \]

  \item[Case 2:] If $1 < n/k \leq a$, then
  \[
  h(k,n,m) = 2^{-1} \sqrt{\frac{\log(m/k+1)}{n}} < 2^{-1} \sqrt{\frac{\log(m/k+1)}{k}}.
  \]

  \item[Case 3:] If $n/k \leq 1$, then since $2^{-k/n} \leq \sqrt{n/k}$ for $k \geq n$, we have
  \[
  h(k,n,m) = 2^{-k/n} \sqrt{\frac{\log(m/k+1)}{n}} \leq \sqrt{\frac{\log(m/k+1)}{k}}.
  \]
\end{description}
Therefore,
\[
h(k,n,m) \lesssim \sqrt{\frac{\log(m/k+1) \log(n/k+1)}{k}} \lesssim \sqrt{\frac{\log m \log n}{k}},
\]
which, together with $\norm{S_\Omega T}_{1 \to \infty} \leq \norm{T}_{1 \to \infty}$, implies
\begin{equation}
\label{eq:sumek1a}
\sum_{k=1}^{n^2-1} \frac{e_k(T(B_1^n), B_{\infty}^n)}{\sqrt{k}}
\lesssim \sum_{k=1}^{n^2-1} \frac{\norm{T}_{1 \to \infty} \sqrt{\log m \log n} }{k}
\leq \norm{T}_{1 \to \infty} \log^{1/2}m \log^{3/2}n.
\end{equation}

For $k \geq n^2$, we compute an upper estimate of the dyadic entropy number using the standard volume argument.
First, we note
\[
e_k(S_\Omega T(B_1^n), B_{\infty}^n) \leq \norm{T}_{1\to\infty} e_k(B_{\infty}^n, B_{\infty}^n).
\]
By the standard volume argument \cite[Lemma~1.7]{pisier1206probabilistic}, we have
\[
N(B_{\infty}^n,\epsilon B_{\infty}^n) \leq (1+2/\epsilon)^n \leq (3/\epsilon)^n,
\]
which implies
\[
e_k(B_{\infty}^n, B_{\infty}^n) \leq 3 \cdot 2^{-(k-1)/n} \leq 2^{-k/(2n)} \leq n/k.
\]
Therefore,
\begin{equation}
\label{eq:sumek1b}
\begin{aligned}
{} & \sum_{k=n^2}^{\infty} \frac{e_k(S_\Omega T(B_1^n), B_{\infty}^n)}{\sqrt{k}} \\
{} & \leq \sum_{k=n^2}^{\infty} \frac{\norm{T}_{1 \to \infty} e_k(B_\infty^n, B_{\infty}^n)}{\sqrt{k}} \\
{} & \leq \sum_{k=n^2}^{\infty} \frac{n \norm{T}_{1 \to \infty}}{ k^{3/2} } \\
{} & \leq 2 \norm{T}_{1 \to \infty}.
\end{aligned}
\end{equation}

Applying (\ref{eq:sumek1a}) and (\ref{eq:sumek1b}) to (\ref{eq:sumek1}) completes the proof.
\end{proof}

\section{Proofs of the Main Results}
\label{sec:proofs_main}

Now, we are ready to prove the main results with Theorem~\ref{thm:ip} in Section~\ref{sec:chaos} and Lemmas~\ref{lemma:est_id} and \ref{lemma:est_fourier} in Section~\ref{sec:key_est}.

\subsection{Proof of Theorem~\ref{thm:bds_rap}}

\begin{proof}[Proof of Theorem~\ref{thm:bds_rap}]

We only prove the first part of Theorem~\ref{thm:bds_rap}.
The proof of the second part follows by symmetry.

Under the assumption of Theorem~\ref{thm:bds_rap}, by Theorem~\ref{thm:gauArip}
$\Psi$ satisfies
\begin{equation}
\label{eq:proof:thm:bds_rap:ripPsi}
\sup_{v \in B_2^n \cap \tGamma_{s_2}} |v^* (\Psi^*\Psi - I_n)v| \leq \delta/2
\end{equation}
except with probability $n^{-\beta_1}$ for an absolute constant $\beta_1 \in \mathbb{N}$.

Since $F$ is unitary, $F \Psi$ has the same distribution to that of $\Psi$.
Let $g_{i,j}$ denote the $(i,j)$th entry of $\sqrt{n} F \Phi$.
Then, $|g_{i,j}|^2$'s are i.i.d. following a Chi-squared distribution with degree of freedom 1.
Since
\[
n \norm{F \Psi}_{1 \to \infty}^2 = \max_{i,j} |g_{i,j}|^2,
\]
by computing the tail distribution of the order statistic, we get
\begin{equation}
\label{eq:proof:thm:bds_rap:maxbnd}
\norm{F \Psi}_{1 \to \infty} \leq c \sqrt{\log n}
\end{equation}
except with probability $n^{-\beta_2}$ for an absolute constant $\beta_2 \in \mathbb{N}$.

We proceed with conditioning on the events in (\ref{eq:proof:thm:bds_rap:ripPsi}) and (\ref{eq:proof:thm:bds_rap:maxbnd}).
In other words, in the remainder of the proof, we will treat $\Psi$ as a deterministic matrix
that satisfies (\ref{eq:proof:thm:bds_rap:ripPsi}) and (\ref{eq:proof:thm:bds_rap:maxbnd}).

\begin{rem}
When $\Psi = I_n$ instead of an i.i.d. Gaussian matrix, we have
\[
\sup_{v \in B_2^n \cap \tGamma_{s_2}} |v^* (\Psi^*\Psi - I_n)v| = 0 \leq \delta/2,
\]
which trivially implies (\ref{eq:proof:thm:bds_rap:ripPsi}).
Therefore, we also obtain Corollary~\ref{cor:bds_rap} once we finish the proof of Theorem~\ref{thm:bds_rap} as follows.
\end{rem}

Define $R_{u,v} \in \mathbb{C}^{m \times n^2}$  and $\xi \in \mathbb{C}^{n^2}$ by
\begin{align*}
R_{u,v} {} & := u^\transpose \otimes \sqrt{\frac{n}{m}} S_\Omega F^* D_{F \Psi v}
\end{align*}
and
\[
\xi := \sqrt{n} (I_n \otimes F) \vect(\Phi).
\]
Then, $R_{u,v} \xi$ satisfies
\[
R_{u,v} \xi = \frac{n}{\sqrt{m}} S_\Omega F^* (F \Psi v \odot F \Phi u)
= \sqrt{\frac{n}{m}} S_\Omega (\Psi v \conv \Phi u)
= \A(u v^\transpose).
\]
Therefore, we have
\[
\langle \hat{u} \hat{v}^\transpose, \A^*\A (u v^\transpose) \rangle = \langle R_{\hat{u},\hat{v}} \xi, R_{u,v} \xi \rangle.
\]

By Lemma~\ref{lemma:isotropyA}, we have
\[
\mathbb{E}_\Phi \langle R_{\hat{u},\hat{v}} \xi, R_{u,v} \xi \rangle
= \langle \hat{u} \hat{v}^\transpose, \mathbb{E}_\Phi \A^*\A (u v^\transpose) \rangle
= \langle \hat{u} \hat{v}^\transpose, u v^\transpose (\Psi^* \Psi)^\transpose \rangle.
\]

By the triangle inequality, we have
\begin{align*}
{} & \big| \langle \hat{u} \hat{v}^\transpose, (\A^*\A - \id) (u v^\transpose) \rangle \big| \\
{} & \leq \big| \langle \hat{u} \hat{v}^\transpose, (\A^*\A - \mathbb{E}_\Phi \A^*\A) (u v^\transpose) \rangle \big|
+ \big| \langle \hat{u} \hat{v}^\transpose, (\mathbb{E}_\Phi \A^*\A - \id) (u v^\transpose) \rangle \big| \\
{} & = | \langle R_{\hat{u},\hat{v}} \xi, R_{u,v} \xi \rangle - \mathbb{E}_\Phi \langle R_{\hat{u},\hat{v}} \xi, R_{u,v} \xi \rangle|
+ \underbrace{ | \hat{u}^* u v^\transpose (\Psi^* \Psi - I_n)^\transpose \overline{\hat{v}} | }_{(*)}.
\end{align*}

By (\ref{eq:proof:thm:bds_rap:ripPsi}), the bias term in the expectation in $(*)$ is upper-bounded by
\begin{align*}
| \hat{u}^* u v^\transpose (\Psi^* \Psi - I_n)^\transpose \overline{\hat{v}} |
\leq \delta/2 \norm{u}_2 \norm{\hat{u}}_2 \norm{v}_2 \norm{\hat{v}}_2
= \delta/2 \norm{\hat{u} \hat{v}^\transpose}_{\mathrm{F}} \norm{u v^\transpose}_{\mathrm{F}}.
\end{align*}

Therefore, it suffices to show
\begin{equation}
\sup_{M, M' \in \Delta} \left| \langle M' \xi, M \xi \rangle - \mathbb{E}_\Phi \langle M' \xi, M \xi \rangle \right| \leq \delta/2,
\label{eq:concenGbil}
\end{equation}
where $\Delta \in \cz^{n \times n}$ is defined by
\begin{equation}
\Delta := \{R_{u,v} : u \in B_2^n \cap \tGamma_{s_1}, \pl v \in B_2^n \cap \tGamma_{s_2} \cap \C_{\mu} \}.
\label{eq:defDelta}
\end{equation}

Since $I_n \otimes F$ is a unitary transform, $\xi \in \cz^{n^2}$ is a Gaussian vector satisfying $\mathbb{E} \xi \xi^* = I_{n^2}$.
The desired concentration of the subgaussian bilinear form in (\ref{eq:concenGbil}) is then derived using Theorem~\ref{thm:ip}.
To apply Theorem~\ref{thm:ip}, we derive upper bounds on
$d_{\mathrm{F}}(\Delta)$, $d_{2\to2}(\Delta)$, and $\gamma_2(\Delta,\norm{\cdot}_{2\to2})$ in the following.

Suppose $R_{u,v} \in \Delta$.
Then, the Frobenius norm of $R_{u,v}$ is written as
\begin{align*}
\norm{R_{u,v}}_{\mathrm{F}}
{} & = \sqrt{\frac{n}{m}} \norm{u}_2 \norm{S_\Omega F^* D_{F \Psi v}}_{\mathrm{F}} \\
{} & \leq \sqrt{\frac{n}{m}} \norm{S_\Omega F^* D_{F \Psi v}}_{\mathrm{F}}.
\end{align*}
In fact, it is upper-bounded by
\begin{align*}
{} & \frac{n}{m} \norm{S_\Omega F^* D_{F \Psi v}}_{\mathrm{F}}^2 \\
{} & = \sum_{k=1}^n \frac{n}{m} \norm{S_\Omega F^* D_{F \Psi v} e_k}_2^2 \\
{} & = \sum_{k=1}^n |e_k^* F \Psi v|^2 \frac{n}{m} \norm{S_\Omega F^* e_k}_2^2 \\
{} & = \sum_{k=1}^n |e_k^* F \Psi v|^2 = \norm{F \Psi v}_2^2 \\
{} & \leq (1+\delta/2) \norm{v}_2^2.
\end{align*}

Meanwhile, the spectral norm of $R_{u,v}$ is upper-bounded by
\begin{align}
\norm{R_{u,v}}_{2\to2}
{} & = \sqrt{\frac{n}{m}} \norm{u}_2 \norm{S_\Omega F^* D_{F \Psi v}}_{2\to2} \nonumber \\
{} & \leq \sqrt{\frac{n}{m}} \norm{u}_2 \norm{S_\Omega F^*}_{2\to2} \norm{F \Psi v}_\infty \label{eq:Ruvsn} \\
{} & \leq \sqrt{\frac{\mu}{m}}, \label{eq:Ruvubsn}
\end{align}
where the last step follows from $v \in C_{\mu}$.

Since $R_{u,v}$ was an arbitrary element of $\Delta$, we deduce
\[
d_{\mathrm{F}}(\Delta) \leq \sqrt{1+\delta/2}.
\]
and
\[
d_{S_\infty}(\Delta) \leq \sqrt{\mu/m}.
\]

Next, by Lemma~\ref{lemma:gamma2}, the last term $\gamma_2(\Delta,\norm{\cdot}_{2\to2})$ is bounded from above by
\[
\gamma_2(\Delta,\norm{\cdot}_{2\to2})
\lesssim \sqrt{\frac{\mu s_1 + s_2}{m}} \log^{5/2} n.
\]

Let $t = \delta/4$. Then, combining upper bounds on $K_1$, $K_2$, and $K_3$ in Theorem~\ref{thm:ip},
we note that there exists an absolute constant $C$ so that $n \geq C \delta^{-2} (\mu s_1 + s_2) \log^5 n$ implies
\[
c_1 K_1 + t \leq \delta/2
\]
and
\[
2 \exp
\left(
-c_2
\min
\left\{
\frac{t^2}{K_2^2}, \frac{t}{K_3}
\right\}
\right)
\leq n^{-\beta_2}
\]
for an absolute constant $\beta_2 \in \mathbb{N}$.
This concludes the proof.
\end{proof}

\begin{lemma}[Isotropy]
Let $\Phi, \Psi \in \cz^{n \times n}$ be independent random matrices whose entries are i.i.d. following $CN(0,1/n)$.
Let $\A$ be defined in (\ref{eq:defcalA}).
Then,
\label{lemma:isotropyA}
\begin{align*}
\mathbb{E}_\Phi \A^* \A(X) {} & = X (\Psi^* \Psi)^\transpose, \\
\mathbb{E}_\Psi \A^* \A(X) {} & = \Phi^* \Phi X.
\end{align*}
\end{lemma}

\begin{proof}[Proof of Lemma~\ref{lemma:isotropyA}]
Note that
\begin{align*}
\langle M_\ell | X \rangle {}
& = \frac{n}{\sqrt{m}} \left\langle \Phi^* F^* \diag(f_{\omega_\ell}) \overline{F} \overline{\Psi}, X \right\rangle \\
{} & = \frac{n}{\sqrt{m}} \pl \Big\langle \sum_{k=1}^n e_k^* f_{\omega_\ell} \Phi^* F^* e_k e_k^* \overline{F} \overline{\Psi}, X \Big\rangle \\
{} & = \frac{n}{\sqrt{m}} \sum_{k=1}^n f_{\omega_\ell}^* e_k (e_k^* F \Psi \otimes e_k^* F \Phi) \vect(X).
\end{align*}

Therefore,
\begin{align*}
{} & \mathbb{E}_\Phi \left( |\vect{M_\ell} \rangle \langle \vect{M_\ell}| \right) \\
{} & = \frac{n^2}{m} \mathbb{E}_\Phi \left[ \sum_{j=1}^n e_j^* f_{\omega_\ell} (\Psi^* F^* e_j \otimes \Phi^* F^* e_j)
\sum_{k=1}^n f_{\omega_\ell}^* e_k (e_k^* F \Psi \otimes e_k^* F \Phi) \right] \\
{} & = \frac{n^2}{m} \mathbb{E}_\Phi \left[ \sum_{j=1}^n \sum_{k=1}^n e_j^* f_{\omega_\ell} f_{\omega_\ell}^* e_k
(\Psi^* F^* e_j e_k^* F \Psi \otimes \Phi^* F^* e_j e_k^* F \Phi) \right] \\
{} & = \frac{n^2}{m} \sum_{j=1}^n |f_{\omega_\ell}^* e_j|^2
\left(\Psi^* F^* e_j e_j^* F \Psi \otimes \frac{1}{n} I_n\right) \\
{} & = \frac{n}{m} \sum_{j=1}^n \left(\Psi^* F^* e_j e_j^* F \Psi \otimes \frac{1}{n} I_n\right)
= \frac{1}{m} \Psi^* \Psi \otimes I_n,
\end{align*}
where the third step follows since
\[
\mathbb{E}_\Phi \Phi^* F^* e_j e_k^* F \Phi =
\begin{cases}
\frac{1}{n} I_n & \text{if $j = k$} \\
0 & \text{otherwise}.
\end{cases}
\]

This implies
\[
\vect\left[ \mathbb{E}_\Phi \left( |M_\ell \rangle \langle M_\ell| X \rangle \right) \right]
= \mathbb{E}_\Phi \left( |\vect{M_\ell} \rangle \langle \vect{M_\ell}| \vect(X) \rangle \right)
= \frac{1}{m} (\Psi^* \Psi \otimes I_n) \vect(X).
\]
Therefore,
\[
\mathbb{E}_\Phi \left( |M_\ell \rangle \langle M_\ell| X \rangle \right)
= \frac{1}{m} X (\Psi^* \Psi)^\transpose.
\]

Finally, we get
\begin{align*}
\mathbb{E}_\Phi \A^* \A(X)
{} & = \sum_{\ell=1}^m |M_\ell \rangle \langle M_\ell| X\rangle = X (\Psi^*\Psi)^\transpose.
\end{align*}

Since
\begin{align*}
{} & \mathbb{E}_\Psi \left( |\vect{M_\ell} \rangle \langle \vect{M_\ell}| \right) \\
{} & = \frac{n^2}{m} \mathbb{E}_\Psi \left[ \sum_{j=1}^n \sum_{k=1}^n e_j^* f_{\omega_\ell} f_{\omega_\ell}^* e_k
(\Psi^* F^* e_j e_k^* F \Psi \otimes \Phi^* F^* e_j e_k^* F \Phi) \right] \\
{} & = \frac{n}{m} \sum_{j=1}^n \left(\frac{1}{n} I_n \otimes \Psi^* F^* e_j e_j^* F \Psi\right)
= \frac{1}{m} I_n \otimes \Phi^*\Phi,
\end{align*}
where the second identity is derived similarly.

\end{proof}

\begin{lemma}
\label{lemma:gamma2}
Let $\Delta$ be defined in (\ref{eq:defDelta}).
Let $\Phi \in \cz^{n \times n}$ be a random matrix whose entries are i.i.d. following $CN(0,1/n)$.
Suppose that $\Psi$ satisfies (\ref{eq:proof:thm:bds_rap:maxbnd}).
Then,
\[
\gamma_2(\Delta,\norm{\cdot}_{2\to2})
\lesssim \sqrt{\frac{\mu s_1 + s_2}{m}} \log^{5/2} n.
\]
\end{lemma}

\begin{proof}[Proof of Lemma~\ref{lemma:gamma2}]

By Dudley's inequality \cite{ledoux1991probability}, the $\gamma_2$ function is bounded from above by
\[
\gamma_2(\Delta,\norm{\cdot}_{2\to2})
\lesssim \int_0^\infty \sqrt{\log N(\Delta, \epsilon B_{S_\infty})} d\epsilon
\]
where $B_{S_\infty}$ denotes the unit ball in the Schatten class $S_\infty$ with the spectral norm $\norm{\cdot}_{2 \to 2}$,
and the covering number $N(\Delta, \epsilon B_{S_\infty})$ is given by
\[
N(\Delta, \epsilon B_{S_\infty}) := \inf
\left\{ k \in \mathbb{N} \pl \left| \pl \exists (y_i)_{i=1}^k \pl \text{s.t.} \pl \Delta \subset \bigcup_{i=1}^k y_i+\epsilon B_{S_\infty} \right. \right\}.
\]

In (\ref{eq:Ruvubsn}), we showed that the spectral norm of $R_{u,v}$ is bounded by $\sqrt{\mu/m}$ for all $R_{u,v} \in \Delta$.
This implies
\[
N(\Delta,\epsilon B_{S_\infty}^n) = 1, \quad \forall \epsilon \geq \sqrt{\mu/m}.
\]
Therefore, the integral reduces to
\[
\int_0^\infty \sqrt{\log N(\Delta,\epsilon B_{S_\infty})} d\epsilon
= \int_0^{\sqrt{\mu/m}} \sqrt{\log N(\Delta,\epsilon B_{S_\infty})} d\epsilon \pl .
\]

We first compute an estimate for the difference.
For $R_{u,v}, R_{u',v'} \in \Delta$, we have
\begin{align*}
& \norm{R_{u,v}-R_{u',v'}}_{2\to2} \\
& \leq \norm{R_{u,v-v'}+R_{u-u',v'}}_{2\to2} \\
& \leq \norm{R_{u,v-v'}}_{2\to2} + \norm{R_{u-u',v'}}_{2\to2} \\
& \leq \sqrt{n/m} \norm{u}_2 \norm{F\Psi(v-v')}_\infty + \sqrt{n/m} \norm{u-u'}_2 \norm{F\Psi v'}_\infty \\
& \leq \sqrt{n/m} \norm{F\Psi(v-v')}_{\infty} + \sqrt{\mu/m} \norm{u-u'}_2 \pl,
\end{align*}
where the third step holds by (\ref{eq:Ruvsn}) and the last step follows from $v \in B_2^n \cap \C_{\mu}$.

Therefore, we get
\[
N(\Delta,\epsilon B_{S_\infty})
\kl N\left(F \Psi(B_2^n \cap \tGamma_{s_2} \cap \C_{\mu}), \frac{\epsilon}{2} \sqrt{\frac{m}{n}} B_{\infty}^n\right)
N\left(B_2^n \cap \tGamma_{s_1},\frac{\epsilon}{2}\sqrt{\frac{m}{\mu}} B_2^n\right),
\]
where the covering numbers in the right-hand-side are defined in $\ell_\infty^n$ and $\ell_2^n$.

Using $\sqrt{a+b}\le \sqrt{a}+\sqrt{b}$, we deduce with a change of variable that
\begin{equation}
\allowdisplaybreaks
\label{eq:lemma_gamma2_ub}
\begin{aligned}
& \int_0^{\sqrt{\mu/m}} \log^{1/2} N(\Delta,\epsilon B_{S_\infty}) d\epsilon  \\
& \leq \int_0^{\sqrt{\mu/m}} \log^{1/2}
N\left( F\Psi(B_2^n \cap \tGamma_{s_2} \cap \C_{\mu}), \frac{\epsilon}{2} \sqrt{\frac{m}{n}} B_{\infty}^n\right) d\epsilon \\
& \quad + \int_0^{c\sqrt{\mu/m}} \log^{1/2} N\left(B_2^n \cap \tGamma_{s_1},\frac{\epsilon}{2}\sqrt{\frac{m}{\mu}}B_2^n\right) d\epsilon \\
& \leq 2 \sqrt{\frac{n}{m}} \underbrace{ \int_0^{\sqrt{\mu/4n}} \log^{1/2}
N( F\Psi(B_2^n \cap \tGamma_{s_2}), \epsilon B_{\infty}^n) d\epsilon }_{ = (*) }\\
& \quad + 2\sqrt{\frac{\mu}{m}} \underbrace{ \int_0^{1/2} \log^{1/2} N(B_2^n \cap \tGamma_{s_1},\epsilon B_2^n) d\epsilon }_{ = (**) } \pl .
\end{aligned}
\end{equation}

By Lemma~\ref{lemma:est_fourier} and (\ref{eq:proof:thm:bds_rap:maxbnd}), an upper bound on $(*)$ is given as
\begin{equation}
\label{eq:sumek_star}
(*) \lesssim c \sqrt{s_2/n} \log^{5/2} n.
\end{equation}

By Lemma~\ref{lemma:est_id}, an upper bound on $(**)$ is given as
\begin{equation}
\label{eq:sumek_dstar}
(**) \lesssim \sqrt{s_1} \log^2 n.
\end{equation}

Plugging (\ref{eq:sumek_star}) and (\ref{eq:sumek_dstar}) into (\ref{eq:lemma_gamma2_ub}) completes the proof.
\qd

\subsection{Proof of Theorem~\ref{thm:bds_rop_cross}}

\begin{proof}[Proof of Theorem~\ref{thm:bds_rop_cross}]
Note that $\langle \hat{u} \hat{v}^\transpose, \A^*\A (u v^\transpose) \rangle$ is rewritten as
\begin{align}
\langle \A(\hat{u} \hat{v}^\transpose), \A (u v^\transpose) \rangle
= \left\langle \sqrt{\frac{n}{m}} S_\Omega (\Phi \hat{u} \conv \Psi \hat{v}),
\pl \sqrt{\frac{n}{m}} S_\Omega (\Phi u \conv \Psi v) \right\rangle. \label{eq:proof:thm:bds_rop_cross:ip}
\end{align}

The random variable in (\ref{eq:proof:thm:bds_rop_cross:ip}) can be understood as a fourth-order Gaussian process indexed by $u$, $\hat{u}$, $v$, and $\hat{v}$. We haven't found relevant results for the suprema of high-order Gaussian processes in the literature. In order to exploit known result for the second-order Gaussian process \cite{krahmer2014suprema}, slightly extended in this paper in Section~\ref{sec:chaos}, we introduce the following trick that lowers the order of the random process using properties of a Gaussian distribution.

Since $\langle u, \hat{u} \rangle = 0$, we have $\mathbb{E} \Phi \hat{u} u^* \Phi^* = 0$.
This implies that $\Phi \hat{u}$ and $\Phi u$ are uncorrelated Gaussian vectors; hence, they are independent.
Let $\widetilde{\Phi}$ be an i.i.d. copy of $\Phi$.
Then, replacing $\Phi \hat{u}$ in (\ref{eq:proof:thm:bds_rop_cross:ip}) by $\widetilde{\Phi} \hat{u}$ does not change the distribution.
Similarly, $\Psi \hat{v}$ and $\Psi v$ are independent;
hence, we can also replace $\Psi v$  in (\ref{eq:proof:thm:bds_rop_cross:ip}) by $\widetilde{\Psi} v$
for an i.i.d. copy $\widetilde{\Psi}$ of $\Psi$ without changing the distribution.
In other words, the inner product in (\ref{eq:proof:thm:bds_rop_cross:ip}) as a random process has the same distribution
to that of the following random variable:
\begin{equation}
\left\langle \sqrt{\frac{n}{m}} S_\Omega (\widetilde{\Phi} \hat{u} \conv \Psi \hat{v}),
\pl \sqrt{\frac{n}{m}} S_\Omega (\Phi u \conv \widetilde{\Psi} v) \right\rangle. \label{eq:proof:thm:bds_rop_cross:ip2}
\end{equation}

Similarly to the proof of Theorem~\ref{thm:bds_rap},
under the assumption of Theorem~\ref{thm:bds_rop_cross},
except with probability $n^{-\beta_1}$,
$\widetilde{\Phi}$ satisfies
\begin{align}
\label{eq:proof:thm:bds_rop_cross:ripPhi}
\sup_{u \in B_2^n \cap \tGamma_{s_1}} |u^* (\widetilde{\Phi}^*\widetilde{\Phi} - I_n)u| \leq \delta/2, \\
\label{eq:proof:thm:bds_rap:maxbndPsi}
\norm{F \widetilde{\Psi}}_{1 \to \infty} \leq c \sqrt{\log n},
\end{align}
and $\widetilde{\Psi}$ satisfies
\begin{align}
\label{eq:proof:thm:bds_rop_cross:ripPsi}
\sup_{v \in B_2^n \cap \tGamma_{s_2}} |v^* (\widetilde{\Psi}^*\widetilde{\Psi} - I_n)v| \leq \delta/2, \\
\label{eq:proof:thm:bds_rap:maxbndPhi}
\norm{F \widetilde{\Phi}}_{1 \to \infty} \leq c \sqrt{\log n},
\end{align}
for absolute constants $c > 0$ and $\beta_1 \in \mathbb{N}$.

We proceed with conditioning on the above events.
Therefore, in the remainder of the proof, $\widetilde{\Phi}$ and $\widetilde{\Psi}$ will be treated as deterministic matrices satisfying (\ref{eq:proof:thm:bds_rop_cross:ripPhi}), (\ref{eq:proof:thm:bds_rap:maxbndPhi}), (\ref{eq:proof:thm:bds_rop_cross:ripPsi}), and (\ref{eq:proof:thm:bds_rap:maxbndPsi}).
Conditioned on $\widetilde{\Phi}$ and $\widetilde{\Psi}$, the order of the random process in (\ref{eq:proof:thm:bds_rop_cross:ip2}) is 2.

Define $R_{u,v} \in \mathbb{C}^{m \times n^2}$ and $\xi_{\mathrm{R}} \in \mathbb{C}^{n^2}$ respectively by
\begin{align*}
R_{u,v} {} & := u^\transpose \otimes \sqrt{\frac{n}{m}} S_\Omega F^* D_{F \widetilde{\Psi} v},
\end{align*}
and
\[
\xi_{\mathrm{R}} := \sqrt{n} (I_n \otimes F) \vect(\Phi).
\]
Then, $R_{u,v} \xi_{\mathrm{R}}$ satisfies
\[
R_{u,v} \xi_{\mathrm{R}} = \frac{n}{\sqrt{m}} S_\Omega F^* (F \widetilde{\Psi} v \odot F \Phi u)
= \sqrt{\frac{n}{m}} S_\Omega (\widetilde{\Psi} v \conv \Phi u).
\]

Define $L_{\hat{u},\hat{v}} \in \mathbb{C}^{m \times n^2}$ and $\xi_{\mathrm{L}} \in \mathbb{C}^{n^2}$ respectively by
\begin{align*}
L_{\hat{u},\hat{v}} {} & := \hat{v}^\transpose \otimes \sqrt{\frac{n}{m}} S_\Omega F^* D_{F \widetilde{\Phi} \hat{u}},
\end{align*}
and
\[
\xi_{\mathrm{L}} := \sqrt{n} (I_n \otimes F) \vect(\Psi).
\]
Then, $L_{\hat{u},\hat{v}} \xi_{\mathrm{L}}$ satisfies
\[
L_{\hat{u},\hat{v}} \xi_{\mathrm{L}} = \frac{n}{\sqrt{m}} S_\Omega F^* (F \Psi v \odot F \widetilde{\Phi} u)
= \sqrt{\frac{n}{m}} S_\Omega (\Psi \hat{v} \conv \widetilde{\Phi} \hat{u}).
\]

Therefore, we have
\begin{align*}
{} & \left\langle \sqrt{\frac{n}{m}} S_\Omega (\widetilde{\Phi} \hat{u} \conv \Psi \hat{v}),
\pl \sqrt{\frac{n}{m}} S_\Omega (\Phi u \conv \widetilde{\Psi} v) \right\rangle \\
{} & = \langle L_{\hat{u},\hat{v}} \xi_{\mathrm{L}}, R_{u,v} \xi_{\mathrm{R}} \rangle \\
{} & = \left\langle
\begin{bmatrix} 0 & L_{\hat{u},\hat{v}} \end{bmatrix}
\begin{bmatrix} \xi_{\mathrm{R}} \\ \xi_{\mathrm{L}} \end{bmatrix}
, \pl
\begin{bmatrix} R_{u,v} & 0 \end{bmatrix}
\begin{bmatrix} \xi_{\mathrm{R}} \\ \xi_{\mathrm{L}} \end{bmatrix}
\right\rangle .
\end{align*}

Note that
\[
\mathbb{E}_{\Phi,\Psi}
\left\langle
\begin{bmatrix} 0 & L_{\hat{u},\hat{v}} \end{bmatrix}
\begin{bmatrix} \xi_{\mathrm{R}} \\ \xi_{\mathrm{L}} \end{bmatrix}
, \pl
\begin{bmatrix} R_{u,v} & 0 \end{bmatrix}
\begin{bmatrix} \xi_{\mathrm{R}} \\ \xi_{\mathrm{L}} \end{bmatrix}
\right\rangle = 0.
\]

Let $\xi := [\xi_{\mathrm{R}}^\transpose, \xi_{\mathrm{L}}^\transpose]^\transpose$.
Then $\xi \in \cz^{2n^2}$ is a Gaussian vector satisfying $\mathbb{E}_{\Phi,\Psi} \xi \xi^* = 0$.

Therefore, it suffices to show
\[
\sup_{\begin{subarray}{c} M \in \Delta_{\mathrm{R}} \\ M' \in \Delta_{\mathrm{L}} \end{subarray}} \left| \langle M' \xi, M \xi \rangle - \mathbb{E}_\Phi \langle M' \xi, M \xi \rangle \right| \leq \delta,
\]
where $\Delta_{\mathrm{R}}, \Delta_{\mathrm{L}} \subset \cz^{m \times 2n^2}$ are respectively defined by
\begin{align*}
\Delta_{\mathrm{R}} {} & := \left\{\begin{bmatrix} R_{u,v} & 0 \end{bmatrix}
:\pl u \in B_2^n \cap \tGamma_{s_1}, \pl v \in B_2^n \cap \tGamma_{s_2} \cap \C_{\mu_2} \right\}, \\
\Delta_{\mathrm{L}} {} & := \left\{\begin{bmatrix} 0 & L_{\hat{u},\hat{v}} \end{bmatrix}
:\pl \hat{u} \in B_2^n \cap \tGamma_{s_1} \cap \C_{\mu_1}, \pl \hat{v} \in B_2^n \cap \tGamma_{s_2} \right\}.
\end{align*}

The desired concentration of the gaussian bilinear form is then derived using Theorem~\ref{thm:ip}.
To apply Theorem~\ref{thm:ip}, we need to compute $d_{\mathrm{F}}(\Delta)$, $d_{2\to2}(\Delta)$, and $\gamma_2(\Delta,\norm{\cdot}_{2\to2})$
of $\Delta_{\mathrm{R}}$ and $\Delta_{\mathrm{L}}$.

Augmenting a matrix by adding zero columns does not change its (Frobenius/spectral) norms.
Therefore, $d_{\mathrm{F}}(\Delta)$, $d_{2\to2}(\Delta)$, and $\gamma_2(\Delta,\norm{\cdot}_{2\to2})$ of $\Delta_{\mathrm{R}}$
are the same to those of $\Delta$ in the proof of Theorem~\ref{thm:bds_rap}, i.e.,
\begin{align*}
d_{\mathrm{F}}(\Delta_{\mathrm{R}}) {} & \leq \sqrt{1+\delta/2}, \\
d_{S_\infty}(\Delta_{\mathrm{R}}) {} & \leq \sqrt{\mu_2/m}, \\
\gamma_2(\Delta_{\mathrm{R}},\norm{\cdot}_{2\to2})
{} & \lesssim \sqrt{\frac{\mu_2 s_1 + s_2}{m}} \log^2 n.
\end{align*}
By symmetry, we also have
\begin{align*}
d_{\mathrm{F}}(\Delta_{\mathrm{L}}) {} & \leq \sqrt{1+\delta/2}, \\
d_{S_\infty}(\Delta_{\mathrm{L}}) {} & \leq \sqrt{\mu_1/m}, \\
\gamma_2(\Delta_{\mathrm{L}},\norm{\cdot}_{2\to2})
{} & \lesssim \sqrt{\frac{s_1 + \mu_1 s_2}{m}} \log^2 n.
\end{align*}

Similarly to the proof of Theorem~\ref{thm:bds_rap},
applying the above bounds to Theorem~\ref{thm:ip} concludes the proof.

\end{proof}

\section{Discussions: Restricted Angle-Preserving Property?}

In fact, the $(\calS,\calS',\delta)$-RAP of $\A$ does not almost preserve the angle between two vectors $w \in \calS$ and $w' \in \calS'$ as we desire. What is preserved is the inner product between $w$ and $w'$ and it is implied that
\[
\left| \frac{\langle \A(w'), \A(w)\rangle}{\norm{\A(w)}_2 \norm{\A(w')}_2} - \frac{\langle w', w\rangle}{\hsnorm{w} \hsnorm{w'}} \right|
\leq \frac{2\delta\sqrt{1+\delta}}{1+\sqrt{1-\delta}},
\quad \forall w \in \calS, \pl \forall w' \in \calS'.
\]
In particular, for $\delta < 1$, we have
\[
\frac{2\delta\sqrt{1+\delta}}{1+\sqrt{1-\delta}} \leq 2\sqrt{2} \pl \delta.
\]
Unlike the conventional $(\calS,\delta)$-RIP of $\A$ that preserves the length of a vector $w \in \calS$ through $\A$,
the strength of the perturbation in the upper bound
does not depend on the input angle $\langle w', w \rangle / \hsnorm{w} \hsnorm{w'}$ but a fixed constant.

On the contrary, every isometry map (without any restriction on the domain)
preserves the inner product and angle, i.e., isometry has an angle-preserving property.
Different implications among such properties due to the restriction on the domain would be of interest for future research.

\section{Conclusion}

We derive a near optimal performance guarantee for the subsampled blind deconvolution problem.
The flat-spectra condition is crucial in obtaining this near optimal performance guarantee.
Mathematically, the structure from the spectral flatness is given as a nonconvex cone,
which motivated various RIP-like properties different from the standard RIP.
In this paper, we derived RIP-like properties in subsampled blind deconvolution at near optimal sample complexity.
Combined with the performance guaranteed derived from these properties in a companion paper \cite{LeeLJB2015},
we show that sparse signals of certain random models are provably reconstructed from samples of their convolution
at near optimal sample complexity.
Extended RIP results on i.i.d. subgaussian and partial Fourier sensing matrices for compressible signals might be of independent interest.

\section*{Acknowledgement}
K. Lee thanks A. Ahmed and F. Krahmer for discussions, which inspired the random dictionary model in this paper.
This work was supported in part by the National Science Foundation under Grants CCF 10-18789, DMS 12-01886, and IIS 14-47879.

\bibliographystyle{IEEEtran}
\bibliography{IEEEabrv,lra,linalg,cs,bdconv,preprint,appl}

\end{document}

%% file: preamble.tex
\usepackage[cmex10]{amsmath}
\usepackage{amsxtra,amssymb,amsthm,amsfonts}
\usepackage{mathtools}
\usepackage[matha,mathx]{mathabx}
\usepackage{mathrsfs}
\usepackage[usenames]{color}
\usepackage{umoline}
\usepackage{datetime}
\usepackage[margin=1in]{geometry}
\usepackage[ruled,boxed,linesnumbered]{algorithm2e}
\usepackage{setspace}
\usepackage{bbm}

\usepackage{enumerate}
\usepackage{graphicx}
\usepackage{caption}
\usepackage{subfig}
\usepackage{amssymb}
\usepackage{array}

\usepackage{tikz}
\usepackage{pgfplots}

\newtheorem{lemma}{Lemma}[section]

\newtheorem{theorem}[lemma]{Theorem}
\newtheorem{cor}[lemma]{Corollary}

\newtheorem{rem}[lemma]{Remark}

\newcommand{\re}{\begin{rem}\rm}
\newcommand{\mar}{\end{rem}}

\newtheorem{defi}[lemma]{Definition}

\newcommand{\fo}{\begin{eqnarray*}}
\newcommand{\mel}{\end{eqnarray*}}

\newcommand{\kl}{\pl \le \pl}
\newcommand{\gl}{\pl \ge \pl}

\newcommand{\ez}{{\mathbb E}}

\newcommand{\cz}{{\mathbb C}}

\newcommand{\id}{\mathrm{id}}

\newcommand{\p}{\hspace{.05cm}}
\newcommand{\pl}{\hspace{.1cm}}

\newcommand{\A}{{\mathcal A}}

\newcommand{\C}{{\mathcal C}}

\newcommand{\calH}{{\mathcal H}}

\newcommand{\calM}{{\mathcal M}}

\newcommand{\calS}{{\mathcal S}}
\newcommand{\calT}{{\mathcal T}}
\newcommand{\tGamma}{{\widetilde{\Gamma}}}

\newcommand{\qd}{\end{proof}\vspace{0.5ex}}

\newcommand\hsnorm[1]{\norm{#1}_{\mathrm{HS}}}

\makeatletter
\newcommand{\opnorm}{\@ifstar\@opnorms\@opnorm}
\newcommand{\@opnorms}[1]{%
  \left|\mkern-1.5mu\left|\mkern-1.5mu\left|
   #1
  \right|\mkern-1.5mu\right|\mkern-1.5mu\right|
}
\newcommand{\@opnorm}[2][]{%
  \mathopen{#1|\mkern-1.5mu#1|\mkern-1.5mu#1|}
  #2
  \mathclose{#1|\mkern-1.5mu#1|\mkern-1.5mu#1|}
}
\makeatother

\newcommand{\norm}[1]{\Vert#1\Vert}

\newcommand{\trace}{\mathrm{tr}}
\newcommand{\transpose}{\top}
\newcommand{\diag}{\mathrm{diag}}
\newcommand{\vect}{\mathrm{vec}}

\newcommand{\conv}{\circledast}